\documentclass[conference]{IEEEtran}
\usepackage{fancyhdr}
\usepackage{ifpdf}
\usepackage{cite}
\usepackage{multirow}
\usepackage{bm}
\usepackage{bbm}
\usepackage{booktabs}
\usepackage{color}
\usepackage{subfigure}
\usepackage{color}

\ifCLASSINFOpdf
\else
  \usepackage{graphicx,xcolor}
\fi

\usepackage{tcolorbox}
\usepackage[cmex10]{amsmath}
\usepackage{url}
\usepackage{adjustbox}
\usepackage{amsmath}
\usepackage{amssymb}
\usepackage{amsfonts}
\usepackage{mathtools}
\usepackage{amsthm}

\theoremstyle{definition}
\theoremstyle{definition}
\newtheorem{dfn}{Definition}

\newtheorem{lemma}{Lemma}
\newtheorem{theorem}{Proposition}
\newtheorem{corollary}{Corollary}
\newtheorem{remark}{Remark}

\usepackage{algorithm,algpseudocode}
\algdef{SE}[DOWHILE]{Do}{DoWhile}{\algorithmicdo}[1]{\algorithmicwhile\ #1}
\algdef{SE}[SUBALG]{Indent}{EndIndent}{}{\algorithmicend\ }%
\algtext*{Indent}
\algtext*{EndIndent}

\DeclareMathOperator*{\argmin}{arg\,min}

\usepackage{tikz}
\usetikzlibrary{arrows,shapes,backgrounds,fit}

\tcbuselibrary{breakable}
\tcbuselibrary{xparse}
\tcbuselibrary{skins}

\DeclareTColorBox{point}{ o }{enhanced, breakable, pad at break*=2mm, toprule at break=0mm, bottomrule at break = 0mm,  before skip = 0.5cm, after skip = 0.5cm, colframe=blue!40!black, colback=blue!5!white, coltitle = white,
IfValueTF={#1}{title={\textbf{#1}}{}}}

\DeclareTColorBox{question}{ o }{enhanced, breakable, pad at break*=2mm, toprule at break=0mm, bottomrule at break = 0mm, before skip = 0.5cm, after skip = 0.5cm, colframe=red!40!black, colback=red!5!white, coltitle = white,
IfValueTF={\textbf{#1}}{title={#1}}{}}

\DeclareTColorBox{ques_tmp}{ o }{enhanced, breakable, pad at break*=2mm, toprule at break=0mm, bottomrule at break = 0mm, before skip = 0.5cm, after skip = 0.5cm, colframe=black!40!black, colback=black!5!white, coltitle = white,
IfValueTF={#1}{title={#1}}{}}

\DeclareTColorBox{opinion}{ o }{enhanced, breakable, pad at break*=2mm, toprule at break=0mm, bottomrule at break = 0mm, before skip = 0.5cm, after skip = 0.5cm, colframe=green!20!black, colback=green!2.4!white, coltitle = white,
IfValueTF={#1}{title={#1}}{}}

\begin{document}
\title{Differentially Private AirComp Federated Learning with Power Adaptation Harnessing Receiver Noise}

\author{
	\IEEEauthorblockN{
		\normalsize Yusuke Koda\IEEEauthorrefmark{2},
		\normalsize Koji Yamamoto\IEEEauthorrefmark{3},
		\normalsize Takayuki Nishio, and
		\normalsize Masahiro Morikura
	}
	\IEEEauthorblockA{
		\small Graduate School of Informatics, Kyoto University,
		Yoshida-honmachi, Sakyo-ku, Kyoto 606-8501, Japan
	}
	\IEEEauthorblockA{
		\IEEEauthorrefmark{2}
		koda@imc.cce.i.kyoto-u.ac.jp
		\IEEEauthorrefmark{3}kyamamot@i.kyoto-u.ac.jp
	}
}

\maketitle
\begin{abstract}
	Over-the-air computation (AirComp)-based federated learning (FL) enables low-latency uploads and the aggregation of machine learning models by exploiting simultaneous co-channel transmission and the resultant waveform superposition.
	This study aims at realizing secure AirComp-based FL against various privacy attacks where malicious central servers infer clients' private data from aggregated global models.
	To this end, a differentially private AirComp-based FL is designed in this study, where the key idea is to harness receiver noise perturbation injected to aggregated global models inherently, thereby preventing the inference of clients' private data.
	However, the variance of the inherent receiver noise is often uncontrollable, which renders the process of injecting an appropriate noise perturbation to achieve a desired privacy level quite challenging.
	Hence, this study designs transmit power control across clients, wherein the received signal level is adjusted intentionally to control the noise perturbation levels effectively, thereby achieving the desired privacy level.
	It is observed that a higher privacy level requires lower transmit power, which indicates the tradeoff between the privacy level and signal-to-noise ratio (SNR).
	To understand this tradeoff more fully, the closed-form expressions of SNR (with respect to the privacy level) are derived, and the tradeoff is analytically demonstrated.
	The analytical results also demonstrate that among the configurable parameters, the number of participating clients is a key parameter that enhances the received SNR under the aforementioned tradeoff.
	The analytical results are validated through numerical evaluations.
\end{abstract}
\IEEEpeerreviewmaketitle
%
%
\section{Introduction}
\label{sec:intro}
To perform intelligent big-data analysis driven using client-generated data from mobile devices while preserving the privacy of the clients, federated learning (FL)\cite{mcmahan2016communication, kairouz2019advances} has attracted considerable attention.
FL is a distributed machine-learning (ML) technique that enables clients to learn a shared ML model collaboratively while keeping their local data on their devices.
Specifically, each client trains an ML model locally and uploads the model parameters instead of uploading their local data to a central server.
The central server aggregates the model parameters and forms a global model.
In FL, clients are not required to expose their private data to the central servers, which is beneficial in terms of client privacy. 
This is unlike conventional MLs that acquire and store the training data in central servers, wherein clients' private data can be exposed to the central servers.

However, despite these promising benefits, FL involves the following two challenges:
First, uploading model parameters requires connectivity of many clients via wireless multiple-access channel with limited radio resources, resulting in considerable upload latencies\cite{mcmahan2016communication, kairouz2019advances}.
Second, despite preventing the exposure of clients' private data, FL still faces the risk of privacy leakage.
This happens because the model parameters uploaded by clients may still be informative\cite{carlini2018secret, melis2019exploiting, fredrikson2015model}, and the clients' private data may be inferred by malicious central servers from the aggregated global model parameters\cite{fredrikson2015model}.
These challenges call for the design of low-latency multiple-access schemes that also preserve privacy of clients in the sense that client private data cannot be inferred by such inference attacks.

Regarding the first issue of upload latency over wireless multiple-access channels, over-the-air computation (AirComp) has emerged as a promising approach\cite{zhu2019broadband,yang2020federated, cao2019optimal, wen2019reduced, jiang2019over}.
AirComp typically computes the desired functions by exploiting the wireless superposition nature yielded from the simultaneous transmission of analog-modulated signals.
This simultaneous transmission leads to dramatic latency reduction when compared to the conventional orthogonal access techniques.
In the FL context, the aggregation of model parameters transmitted by clients can be performed using AirComp.
Typically, these existing studies aim at minimizing the noise perturbation injected to the computation results of global model parameters.
For example, this goal can be achieved by enhancing the received signal-to-noise ratio (SNR) using user scheduling \cite{zhu2019broadband} or by minimizing the mean squared error (MSE) between the desired computation results and received symbols through transmit power control\cite{cao2019optimal}, analog beamforming\cite{yang2020federated}, multiple-input--multiple-output (MIMO) beamforming\cite{wen2019reduced}, or phase shifting by reflecting surfaces\cite{jiang2019over}.
However, the minimization of noise perturbation and the consequent accurate estimation of the desired computation results still include risks from the aforementioned inference attacks performed by malicious central servers in base stations (BSs) or access points in the AirComp-based FL.

To tackle these two challenges jointly, we aim to design an AirComp-based FL that is secure against the aforementioned inference attacks.
This goal is achieved by harnessing the perturbation from the inherent receiver noises, instead of minimizing the noise perturbations, thereby preserving differential privacy\cite{dwork2014algorithmic}. 
In a nutshell, differential privacy is a privacy notion quantified by the difference in possible outcomes of data aggregations (e.g., calculating an arithmetic mean) performed with or without each individual's data. 
The smaller difference indicates that each individual's contribution is smaller, which makes it harder for adversaries to infer data of each individual.
Hence, the smaller difference is interpreted as a higher privacy level.
In the FL context, designing a differentially private model aggregation with a higher privacy level makes it harder for malicious model aggregators to infer clients local models, resulting in the privacy protection from the aforementioned inference attacks\cite{mcmahan2017learning, geyer2017differentially, hao2019towards, truex2019hybrid}.
Recent studies designed a differentially private FL in non-AirComp by injecting artificial noises to aggregated global model parameters.
However, unlike these studies, our objective is to design an AirComp-based FL that preserves differential privacy by using the inherent receiver noises.
Therein, the differential privacy level depends on the configurable wireless parameters, e.g., transmit powers, and hence, designing such parameters is necessary, which is an essential difference from designing differential private FL in non-AirComp.

The main contributions of this paper are as follows:

\begin{itemize}
	\item We design AirComp-based FL that achieves data security against the aforementioned privacy inference attacks, based on the following two ideas.
	The first idea is preserving a higher level of differential privacy\cite{dwork2014algorithmic} by harnessing the perturbation from receiver noises inherently injected into aggregated global models.
	However, utilizing receiver noises poses the challenge of injecting an appropriate noise perturbation level to achieve a desired privacy level, because the variance of receiver noises is often uncontrollable.
	Motivated by this, the second idea is to intentionally adjust the received signal level by controlling the transmit powers, thereby controlling the perturbation levels effectively so that the desired privacy level is achieved.
	Through numerical evaluations, we demonstrate that the designed transmit power control achieves a higher privacy level compared to a conventional power control\cite{zhu2019broadband} while exhibiting comparable training performance.
	\item To obtain a fuller understanding of the performance of the designed AirComp-based FL, we derive a closed-form expression that presents the relationship between the essential performance metrics: received SNR and differential privacy level.
	The analytical results demonstrate the following two facts: (i) there is a challenging tradeoff between the received SNR and differential privacy level; (ii) under a constraint wherein a higher privacy level is desired, the number of participating clients will be a major configurable parameter that can enhance the received SNR.
	The analytical results are verified through numerical evaluations.
\end{itemize}

The remainder of this paper is organized as follows:
In Section~\ref{sec:system_model}, we provide a system model.
In Section~\ref{sec:tp_design}, we design a transmit power control scheme to preserve the differential privacy with a target privacy level in the AirComp-based FL model aggregation.
In Section~\ref{sec:snr_analysis}, we analytically derive the received SNR and investigate the tradeoffs between the received SNR and privacy level.
In Section~\ref{sec:numerical_evaluation}, we present a numerical evaluation to verify our analytical results.
In Section~\ref{sec:conclusion}, we present our concluding remarks.

\section{System Model}
\label{sec:system_model}

\begin{figure}[t]
	\centering
	\includegraphics[width=\columnwidth]{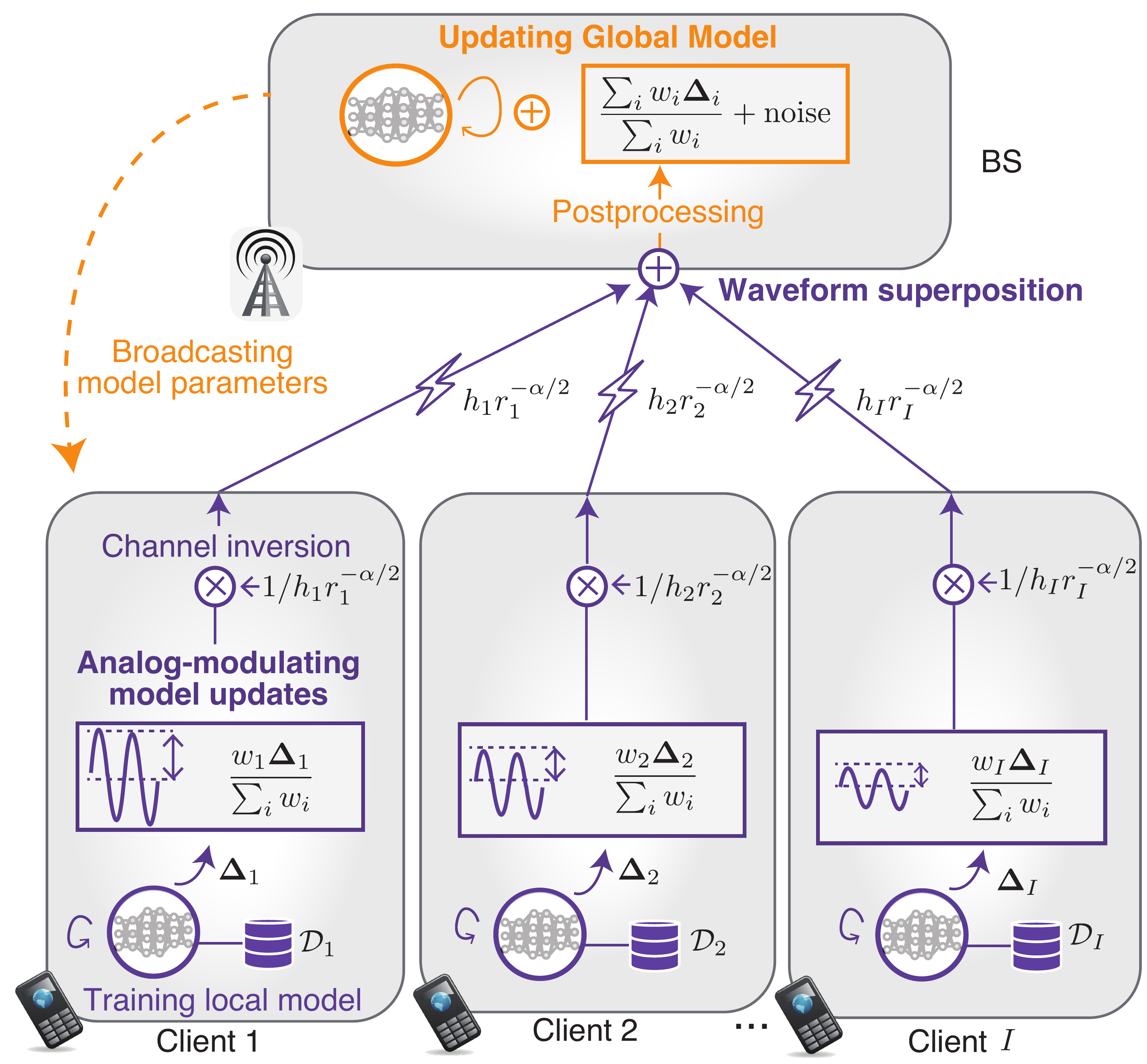}
	\caption{
		AirComp-based FL\cite{zhu2019broadband, yang2020federated}. 
		The weighted average of model updates $\bm{\Delta}_i$ for $i = 1, 2, \dots I$ is computed by exploiting the superposition of analog-modulated wireless signals.
		Prior to transmitting amplitude-modulated model updates, wireless channel inversion is performed at the client side.
	}
	\label{fig:over_the_air_computation}
	\vspace{-1.5em}
\end{figure}

\subsection{Federated Learning}
We consider an FL system comprising $I$ clients and one BS.
A shared model represented by parameters $\bm{\theta}$ is trained cooperatively across all clients with their own local datasets, each of which is denoted by $\mathcal{D}_i$, where $i\in\mathcal{I}\coloneqq\{1, \dots, I\}$ denotes the index of clients
\cite{mcmahan2016communication}.
The objective of this learning is to train the parameters $\bm{\theta}$ to minimize the model error, termed loss function, thereby, obtaining a good approximation for the target labels:
\begin{align}
	\label{eq:learning_obj}
	\bm{\theta}^{\star} = \argmin_{\bm{\theta}}\frac{1}{|\mathcal{D}|}\sum_{(x_j, y_j)\in \mathcal{D}}f(x_j, y_j; \bm{\theta}),
\end{align}
where $\mathcal{D}\coloneqq \cup_{i}\mathcal{D}_i$.
In \eqref{eq:learning_obj}, $f(x_j, y_j; \bm{\theta})$ denotes the loss function with respect to the samples $(x_j, y_j)\in\mathcal{D}$ quantifying the error between the model output and the target label $y_j$ with respect to the input sample $x_j$.
The straightforward approach to learn $\bm{\theta}^\star$ is to gather the local datasets in the BS and minimize $\sum_{(x_j, y_j)\in \mathcal{D}}f(x_j, y_j; \bm{\theta})/{|\mathcal{D}|}$ with respect to the parameters $\bm{\theta}$ using, for example, the stochastic gradient decent method.
However, this exposes the local datasets to the BS, which is not desirable owing to privacy concerns.

Alternatively, FL can learn $\bm{\theta}^\star$ in a distributed manner keeping the local datasets in each client\cite{mcmahan2016communication}.
Specifically, FL minimizes the \textit{local loss function} defined as $F_j(\bm{\theta}) \coloneqq \sum_{(x_j, y_j)\in \mathcal{D}_i}f(x_j, y_j; \bm{\theta}) / {|\mathcal{D}_i|}$ and integrates the learned parameter updates  by, for example, taking a weighted average of the parameters updates.
The detailed procedure is as follows:
First, the BS distributes the shared parameters $\bm{\theta}_{\mathrm{global}}$ across all clients.
Subsequently, each client computes the parameter updates $\bm{\Delta}_{i}$ to minimize the local loss function.
Finally, the BS integrates the parameter updates uploaded by the clients by, for example, taking a weighted average as\cite{mcmahan2016communication}:
\begin{align}
	\bm{\theta}_{\mathrm{global}} \leftarrow \bm{\theta}_{\mathrm{global}} + \frac{\sum_{i\in \mathcal{I}} w_i \bm{\Delta}_{i}}{\sum_{i\in \mathcal{I}}w_i},
\end{align}
where $w_i\in\mathbb{R}^{+}$ for $i = 1, 2, \dots, I$ denotes the weights associated to client~$i$.
Typically, the weights are set to be the data size of each client, i.e., $|\mathcal{D}_i|$ for $i = 1, 2, \dots, I$\cite{mcmahan2016communication}.
This procedure is termed \textit{round} and is iterated until the terminate conditions are satisfied, e.g., model performance converges or a predefined number of rounds has been reached.

\subsection{AirComp for FL Model Aggregation}
\label{system_model:over_the_air_computation}
In AirComp-based FL, the BS performs model aggregation by exploiting the superposition nature of the wireless signals transmitted by all clients, as illustrated in Fig.~\ref{fig:over_the_air_computation}.
Let the local update computed by client $i$ in a typical round be denoted as $\bm{\Delta}_i = \bigl[\Delta^{(1)}_i, \dots, \Delta^{(D)}_i\bigr]^{\mathrm{T}}$.
We consider that transmission time is slotted as in time division multiple access (TDMA) channel in a cellular network, and that each element of the local update is transmitted in each time slot.
The clients modulate the symbol $s^{(d)}_i\coloneqq w_i \Delta^{(d)}_i / \sum_{i\in\mathcal{I}} w_i$ for $d\in\{1, 2,  \dots, D\}$ with analog-amplitude modulation, wherein modulated symbols are constant over a time slot.
We also consider that all clients perform a slot-level synchronization through a synchronization channel as in timing advance in LTE systems\cite{timing_advance}.
Therein, all clients transmit the modulated signals simultaneously to exploit the superposition nature of the wireless signals.
The modulated signal in the client $i$ is given by
\begin{align}
	\label{eq:transmitted}
	y_i^{(d)} = \sqrt{2}s^{(d)}_ib^{(d)}_i\cos(2\pi ft),
\end{align}
where $f$ denotes the central frequency of the modulated signal and $b^{(d)}_i$ denotes the power control policies to prohibit the transmit power from exceeding the predefined maximum value of $P_0$.
Hereinafter, we consider the computation for one time slot and skip the notation indicating the time slot $(d)$, for ease of notation.

We consider that the modulated signal is transmitted via a Rayleigh fading and additive white Gaussian noise channel, and that it is detected by a matched-filter procedure\cite{goldsmith} for the transmitted signal in \eqref{eq:transmitted}.
Given the distance between the client~$i$ and BS, $r_i$, the product of the transmit and receive antenna gains, $G$, and the path loss exponent $\alpha$, the received symbol is given by 
\begin{align}
	\label{eq:received_signal}
	r = \sum_{i\in \mathcal{I}} \sqrt{\beta}r^{-{\alpha}/{2}}_i \sqrt{G} h_i  b_i s_i + n,
\end{align}
where $\beta$ is the path loss for a reference unit distance, $n\sim\mathcal{CN}(0, \sigma^2_{\mathrm{n}})$ is the receiver noise, and  $h_i\sim \mathcal{CN}(0, 1)$ is the fading channel gain.

The client-side transmission power strategy $b_i$ is given such that we obtain the desired computation $\sum_{i\in\mathcal{I}}s_i$ from the received symbol.
The solution is the channel inversion preprocessing with the maximum power constraint $|b_i s_i|^2\leq P_0$\cite{zhu2019broadband}.
In the channel inversion, $b_i$ is set as $\sqrt{\rho}/(r_i^{-\alpha/2}h_i)$, where $\rho$ is constant for all clients and is termed the power-scaling factor.
Considering that the maximum transmit power constraint should be satisfied across all clients, the power-scaling factor $\rho$ should be\cite{zhu2019broadband}
\begin{align}
	\label{eq:channel_inversion}
	\rho \leq \min_{i\in\mathcal{I}} \frac{r^{-\alpha}_i |h_i|^2}{|s_i|^2}P_0.
\end{align}
Given the channel inversion, the received symbol at the BS results in
\begin{align}
	\label{eq:receive_signal2}
	r = \sqrt{G\beta\rho}\sum_{i\in \mathcal{I}} s_i + n_0.
\end{align}
We take the real part of the received symbol and multiply the inversion of $\sqrt{G\beta\rho}$ to obtain the desired model update aggregation $\sum_{i\in\mathcal{I}}s_i$.
The computation results are as follows:
\begin{align}
	\label{eq:computation_result_original}
	\sum_{i\in\mathcal{I}} s_i + \frac{1}{\sqrt{G\beta\rho}}n_{\mathrm{I}},
\end{align}
where $n_{\mathrm{I}}\sim\mathcal{N}(0, \sigma^2_{\mathrm{n}}/2)$ is the real part of the receiver noise.

\section{Transmit Power Control in AirComp FL for Differential Privacy}
\label{sec:tp_design}
We design the transmit power control for an AirComp-based FL to achieve differential privacy.
Prior to providing details, we note the definition of differential privacy.
Differential privacy is defined to characterize randomized mechanisms that output a desired value computed over a dataset perturbed with a random noise.
The definition of differential privacy is as follows:
\begin{dfn}
	(Differential Privacy\cite{dwork2014algorithmic})
	A randomized mechanism $\mathcal{M}$ is $(\epsilon, \delta)$-differentially private if for any pair of adjacent dataset \footnote{In the FL context, the notion of ``adjacent'' means that one dataset can be formed by adding or removing all examples associated with a single client from the other dataset\cite{mcmahan2017learning}.}$d$ and $d'$ and any sort of possible outcome $\mathcal{S}\subseteq\mathrm{Range}(\mathcal{M})$, we obtain
	\begin{align}
		\label{eq:DP}
		\mathbb{P}(\mathcal{M}(d)\in \mathcal{S}) \leq \mathrm{e}^{\epsilon}\mathbb{P}(\mathcal{M}(d')\in \mathcal{S}) + \delta.
	\end{align}
\end{dfn}
\noindent In the above definition, $\mathrm{Range}(\mathcal{M})$ is the set of all possible outcomes of $\mathcal{M}$.
Note that the values $\epsilon$ and $\delta$ represent the similarity in the distribution of the outcomes of the randomized mechanisms performed over the datasets $d$ and $d'$, and are interpreted as a privacy level\cite{dwork2014algorithmic}.
Lower $\epsilon$ and $\delta$ account for a higher privacy level.

The computation in \eqref{eq:computation_result_original} is also a randomized mechanism, and the noise perturbation can be controlled by the power-scaling factor $\rho$.
Hence, an appropriate design of the power-scaling factor realizes the $(\epsilon, \delta)$-differential privacy.
We set the power-scaling factor so that differential privacy with the privacy level $\epsilon$ and $\delta$ can be preserved in the following discussion.
Note that as the power-scaling factor and transmit power are proportional to each other, we use ``setting the power scaling factor'' and ``transmit power control'' interchangeably hereinafter.

\vspace{-.5em}
\subsection{Update Clipping}
In the computation discussed above, there is no limit to the \textit{query sensitivity}, which is crucial to preserve the differential privacy\cite{dwork2014algorithmic}.
The sensitivity is generally referred to as the maximum effect of the attendance of one entity holding data on the desired computation results.
In the case of the computation in \eqref{eq:computation_result_original}, the query sensitivity is given by\cite{mcmahan2017learning}
	\begin{align}
		\label{eq:sensitivity}
		\max_{k\in\mathcal{I}}\left|\sum_{i\in\mathcal{I}}s_i - \sum_{i\in\mathcal{I}\setminus\{k\}}s_i\right|.
	\end{align}

	A solution to this issue is the application of \textit{update clipping}, which was proposed in \cite{mcmahan2017learning} for non-AirComp FL.
	In update clipping, the weighted local update in each client $w_i\Delta_i/\sum_{i\in\mathcal{I}} w_i$  is bounded by a threshold $S > 0$, as follows:
\begin{align}
	\label{eq:clipping}
	s_i = \frac{w_i\Delta_{i}}{\sum_{i \in \mathcal{I}} w_i} \min\left\{1, \frac{\sum_{i \in \mathcal{I}}w_i}{w_i|\Delta_{i}|}S\right\}.
\end{align}
In update clipping, the query sensitivity is also bounded by $S$, as follows:
The sensitivity in \eqref{eq:sensitivity} is equivalent to $\max_{k}|s_k|$, which is obviously less or equal to $S$ in the clipping in \eqref{eq:clipping}.

\subsection{Power-Scaling Factor Constraint for Differential Privacy}
Based on the aforementioned update clipping, we adaptively choose the power-scaling factor $\rho$ with the objective of the computation in \eqref{eq:computation_result_original} being $(\epsilon, \delta)$-differentially private with the target privacy level of $\epsilon$ and $\delta$.
First, we provide the condition that the computation in \eqref{eq:computation_result_original} is $(\epsilon, \delta)$-differentially private, as follows:\vspace{-.5em}
\begin{lemma}
		\label{lemma:diff_priv_condition}
		(Power-scaling factor constraint for differential privacy) Given the target privacy level $\epsilon$ and $\delta$, the computation in \eqref{eq:computation_result_original} is $(\epsilon, \delta)$-differentially private if
		\begin{align}
			\label{eq:diff_priv_condition}
			\rho\leq \frac{\sigma^2_{\mathrm{n}}}{4\,G\beta}\frac{\epsilon^2}{S^2\,\ln(1.25/\delta)}.
		\end{align}
	\end{lemma}
	\begin{proof}
		From the reproducibility of the Gaussian distribution, the second term in \eqref{eq:computation_result_original} is also a random variable following a Gaussian distribution with standard deviation of ${\sigma_n}/\sqrt{2G\beta\rho}$.
	Given the sensitivity of the query bounded by $S$, the computation in \eqref{eq:computation_result_original} is $(\epsilon, \delta)$-differentially private if\cite{dwork2014algorithmic}:
	\begin{align*}
		\frac{\sigma_n}{\sqrt{2G\beta\rho}} \geq \frac{S\sqrt{2\ln (1.25/\delta)}}{\epsilon}.
	\end{align*}
	By solving for $\rho$, we obtain \eqref{eq:diff_priv_condition}.
	\end{proof}

	By setting the power-scaling factor $\rho$ to satisfy both \eqref{eq:channel_inversion} and \eqref{eq:diff_priv_condition}, we can achieve $(\epsilon, \delta)$-differential privacy in the computation in \eqref{eq:computation_result_original}.
	Note that the scaling factor is set to be the maximum value that satisfies these two constraints, to enhance the received SNR.
	\begin{theorem}
		(Transmit power control preserving differential privacy)
		If we set the power-scaling factor  $\rho$ as $\rho^{\star}$ given by:
		\begin{align}
			\label{eq:priv_channel_inversion}
			\rho^{\star} = P_0\min\!\left\{\min_{i\in\mathcal{I}} \frac{r^{-\alpha}_i |h_i|^2}{|s_i|^2}, \frac{\sigma^2_{\mathrm{n}}}{4\,G\beta P_0}\frac{\epsilon^2}{S^2\,\ln(1.25/\delta)}\right\},
		\end{align}
		then, the computation in \eqref{eq:computation_result_original} is $(\epsilon, \delta)$-differentially private.
		\begin{proof}
			The power-scaling factor $\rho^{\star}$ satisfies the constraint for differential privacy \eqref{eq:diff_priv_condition}; therefore, the computation in \eqref{eq:computation_result_original} with $\rho^{\star}$ is $(\epsilon, \delta)$-differentially private.
		\end{proof}
	\end{theorem}

\subsection{Modified Transmit Power Control Anonymizing Client Updates}

While setting the power-scaling factor as in \eqref{eq:priv_channel_inversion} to perform channel inversion, each client requires the information on the client contribution $s_i$ for all clients, because the power-scaling factor in \eqref{eq:priv_channel_inversion} is set according to $\min_{i}(r^{\alpha}_i |h_i|^2/|s_i|^2)$.
However, this is contrary to our objective of anonymizing the clients' contributions.

Motivated by this, we modify the strategy for setting the scaling factor satisfying the maximum power constraint \eqref{eq:channel_inversion} such that information on $s_i$ is not required.
In update clipping, we have $r^{\alpha}_i|h_i|^2/|s_i|^2\geq r^{\alpha}_i|h_i|^2/S^2$ for all clients, and then $\min_{i}r^{\alpha}_i|h_i|^2/|s_i|^2\geq \min_{i}r^{\alpha}_i|h_i|^2/S^2$.
Hence, if we set $\rho$ such that the following condition is satisfied, the constraint \eqref{eq:channel_inversion} will be satisfied for all clients:
\begin{align}
	\label{eq:priv_channel_inv}
	\rho \leq \min_{i\in\mathcal{I}} \frac{r^{-\alpha}_i |h_i|^2}{S^2}P_0.
\end{align}

From \eqref{eq:priv_channel_inv} and \eqref{eq:diff_priv_condition}, we design the channel inversion strategy such that the information on $s_i$ for all clients is not required.
We now have the following corollary:
\begin{corollary}
	(Transmit power control preserving differential privacy without information on $s_i$ for all clients)
	If we set the power-scaling factor $\rho$ as $\rho^{\star\star}$, given by
	\begin{align}
		\label{eq:priv_channel_inversion2}
		\rho^{\star\star} = \frac{P_0}{S^2}\min\!\left\{\min_{i\in\mathcal{I}} {r^{-\alpha}_i |h_i|^2}, \frac{\sigma^2_{\mathrm{n}}}{4\,G\beta P_0}\frac{\epsilon^2}{\,\ln(1.25/\delta)}\right\},
	\end{align}
	then, the computation in \eqref{eq:computation_result_original} will be $(\epsilon, \delta)$-differentially private.
	
	\begin{proof}
		The power-scaling factor $\rho^{\star\star}$ satisfies the constraint for differential privacy \eqref{eq:diff_priv_condition}; therefore, the computation in \eqref{eq:computation_result_original} with the strategy $\rho^{\star\star}$ will be $(\epsilon, \delta)$-differentially private.	
	\end{proof}
\end{corollary}

\section{Analysis of Tradeoff between SNR and Privacy Level}
\label{sec:snr_analysis}
In the designed transmit power control, there is a tradeoff between the received SNR and privacy level $(\epsilon, \delta)$, and we derive this tradeoff analytically.
As a higher received SNR generally leads to a better performance of FL models\cite{zhu2019broadband}, this analytical result provides an insight into the difficulty of enhancing the model performance learned in the designed AirComp while preserving the differential privacy with a higher privacy level.

The tradeoff between the received SNR and privacy level can be obtained by deriving the SNR bound in the designed transmit power control, as follows:
\begin{theorem}
	(SNR-privacy-level tradeoff)
	The received SNR in the designed transmit power control is bounded by
	\begin{align}
		\label{eq:snr_bound}
		\mathit{SNR} &\leq \frac{G\beta I^2P_0}{\sum_{i \in \mathcal{I}}r^{\alpha}_i\sigma^2_{\mathrm{n}}}
		\left[1 - \exp\!\left({- \frac{\sum_{i \in \mathcal{I}}r^{\alpha}_i\sigma^2_{\mathrm{n}}}{4\,G\beta P_0}\frac{\epsilon^2}{\ln(1.25/\delta)}}\right)\right].
	\end{align}
\end{theorem}
\begin{proof}
	See Appendix.
\end{proof}
	\begin{remark}
		The SNR bound in \eqref{eq:snr_bound} decreases monotonously as the target privacy level increases (i.e., $\epsilon$ and $\delta$ decrease).
		This indicates the tradeoff between the received SNR and privacy level exactly.

		It is also remarkable that the number of clients $I$ is the most important factor for enhancing the received SNR while preserving a higher privacy level, i.e., smaller values of $\epsilon$ and $\delta$.
		This fact can be expressed using the first-order Taylor approximation of the upper bound for the received SNR in \eqref{eq:snr_bound}.
		Let the upper bound of the received SNR be denoted by $\mathit{SNR}_{\mathrm{b}}$.
		For $\epsilon^2 / \ln(1.25/\delta) \ll 4G\beta P_0/\sum_{i\in\mathcal{I}} r^{\alpha}_{i}\delta^2_{\mathrm{n}}$, the upper bound of the received SNR can be approximated by:
		\begin{align}
			\label{eq:snr_bound_simple}
			\mathit{SNR}_{\mathrm{b}} & = 
			\frac{G\beta I^2P_0}{\sum_{i \in \mathcal{I}}r^{\alpha}_i\sigma^2_{\mathrm{n}}}
			\!\left[1 - \exp\!\left({- \frac{\sum_{i \in \mathcal{I}}r^{\alpha}_i\sigma^2_{\mathrm{n}}}{4\,G\beta P_0}\frac{\epsilon^2}{\ln(1.25/\delta)}}\right)\right]\nonumber\\
			& \approx \frac{G\beta I^2P_0}{\sum_{i \in \mathcal{I}}r^{\alpha}_i\sigma^2_{\mathrm{n}}}\frac{\sum_{i \in \mathcal{I}}r^{\alpha}_i\sigma^2_{\mathrm{n}}}{4\,G\beta P_0}\frac{\epsilon^2}{\ln(1.25/\delta)}
			\nonumber\\& = \frac{I^2}{4}\frac{\epsilon^2}{\ln(1.25/\delta)}.
		\end{align}
		Clearly, the received SNR depends only on the number of clients.
		Hence, the number of clients is a key factor for realizing a higher SNR while preserving a higher privacy level.
	\end{remark}

\section{Numerical Evaluation}
\label{sec:numerical_evaluation}
\subsection{Settings}
\noindent \textbf{Clients and BS setting:} The evaluation is performed under the condition that the clients are placed depart from a BS, at a distance of 100\,m, i.e., $r_i = 100$\,m for $i\in\{1, \dots, I\}$.
Both the clients and BS are equipped with omni-antennas, and hence, $G = 0$\,dBi.
The central frequency is considered to be 5.0\,GHz.
The variance of the receiver noise is $\sigma_{\mathrm{n}}^2=-60$\,dBm.
The path loss exponent is $\alpha=2$.
The path loss for a reference unit distance $\beta$ is $-46$\,dB.
\vspace{.3em}
\\
\noindent \textbf{Data set:} We use the well-known MNIST dataset that consists of 10 categories of hand-written digits ranging from ``0'' to ``9''.
In the MNIST dataset, the total number of training data samples $|\mathcal{D}|$ is 60000.
The training data samples are randomly partitioned into $I$ equal shares, wherein each client holds $60000/I$ training data samples.\vspace{.3em}\\
\noindent \textbf{Training details:} 
The classifier model is implemented using two-layer fully connected neural networks.
The number of units in each layer is 512.
The last fully connected layer is followed by a softmax output layer.
The training is performed by minimizing the loss function that is a  categorical cross-entropy in this setting.
As the optimizer, we use an Adam optimizer\cite{sutskever2013importance} with the learning rate of $1.0\times 10^{-3}$, decaying rate parameters $\beta_1=0.9$ and $\beta_2=0.999$, and batch size 32.
The model aggregation is performed per $20$ epochs for stochastic gradient descent steps.
The clipping threshold $S$ is set as $5.0\times 10^{-5}$.
The privacy level $\delta$ is set as $0.1$.

\subsection{Results}
\begin{figure}[t]
\centering
\includegraphics[width=0.85\columnwidth]{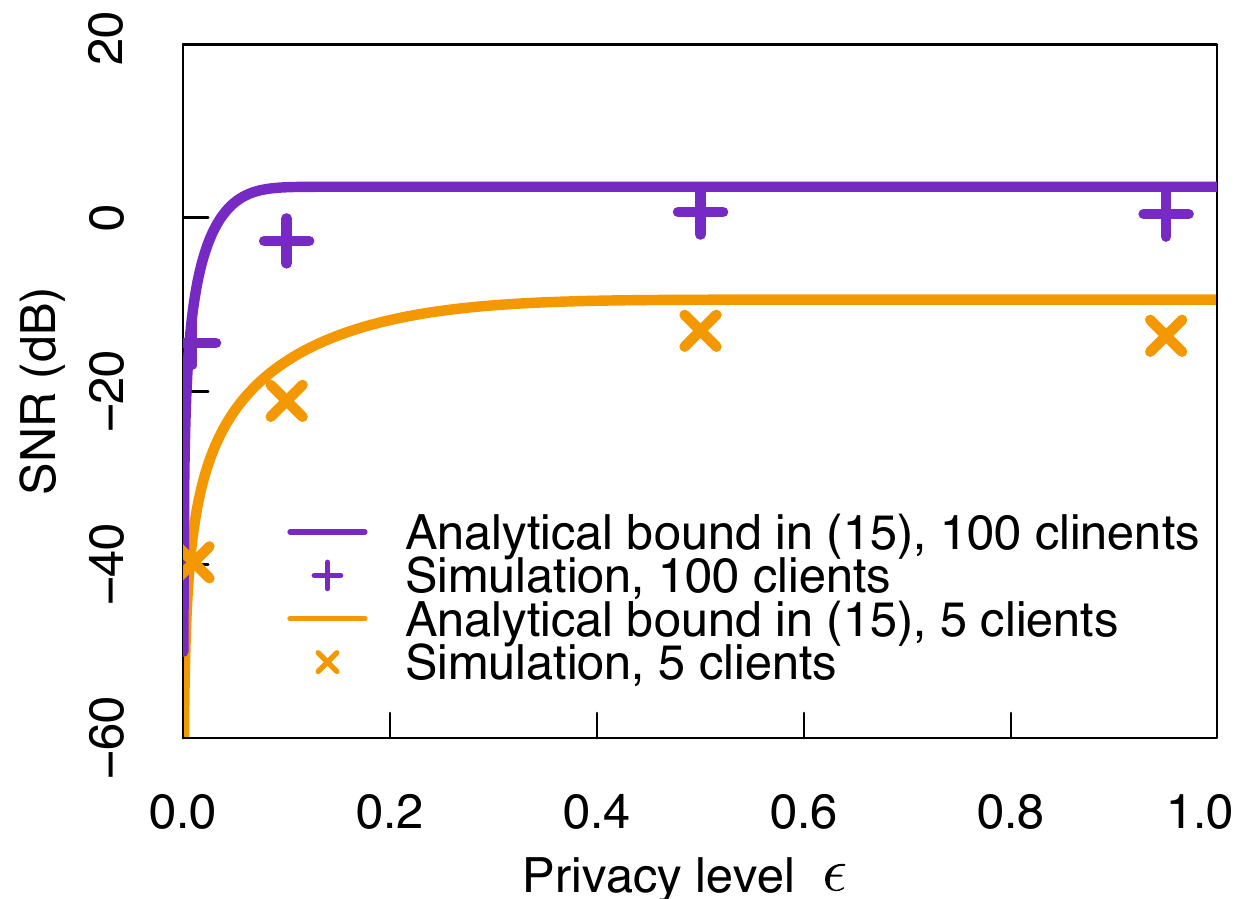}\vspace{-.8em}
\caption{
	Comparison of analytical results of SNR bound in \eqref{eq:snr_bound} and SNR measured from the simulations for various $\epsilon$ with maximum transmit power of $P_0 = 10$\,dBm.
	A lower $\epsilon$ indicates a higher differential privacy level; hence, this figure demonstrates the tradeoff between the received SNR and differential privacy level.
}
\label{fig:theory_simulation}
\vspace{-1.7em}
\end{figure}
\noindent \textbf{Validity of analytical result for received SNR bound:} We validate the analytical results of the upper bound of the received SNR. 
Fig.~\ref{fig:theory_simulation} shows the SNR bound in \eqref{eq:snr_bound} for the number of clients equal to five and 100 and received SNR measured in the simulation for $\epsilon = 0.01, 0.1, 0.5, 0.95$.
The maximum transmit power is set as $P_0 = 10$\,dBm.
Fig.~\ref{fig:theory_simulation} validates the analysis in the sense that the derived bound for the received SNR is indeed an upper bound of the received SNR measured in the simulation.
Moreover, Fig.~\ref{fig:theory_simulation} demonstrates that a lower $\epsilon$ results in a lower received SNR, and this tendency coincides with the analytical upper bound of the received SNR.
This result also validates our statement: there is a tradeoff between the received SNR and differential privacy level.

We validate the analytical results in \eqref{eq:snr_bound_simple}, i.e., for a higher target privacy level, the number of clients participating in the computation has a significant impact on the received SNR, relative to other wireless factors.
Fig.~\ref{fig:theory_simulation_lower_epsilon} shows the received SNR in the simulations, along with the analytical results in \eqref{eq:snr_bound_simple}, under the condition of $\epsilon=0.01$.
As an example, Fig.~\ref{fig:theory_simulation_lower_epsilon} provides the received SNR for the two maximum transmit powers $P_0$, i.e., 10\,dBm and 30\,dBm.
We can see that, as the number of clients increases, the received SNR increases, along with the analytical SNR bound in \eqref{eq:snr_bound_simple}.
When compared to the increase in the number of clients, the increase in the maximum transmit power does not have an impact on the received SNR.
These facts validate our statement: for a higher target privacy level, the number of clients is a key factor that enhances the received SNR.

\begin{figure}[t]
	\centering
	\includegraphics[width=0.85\columnwidth]{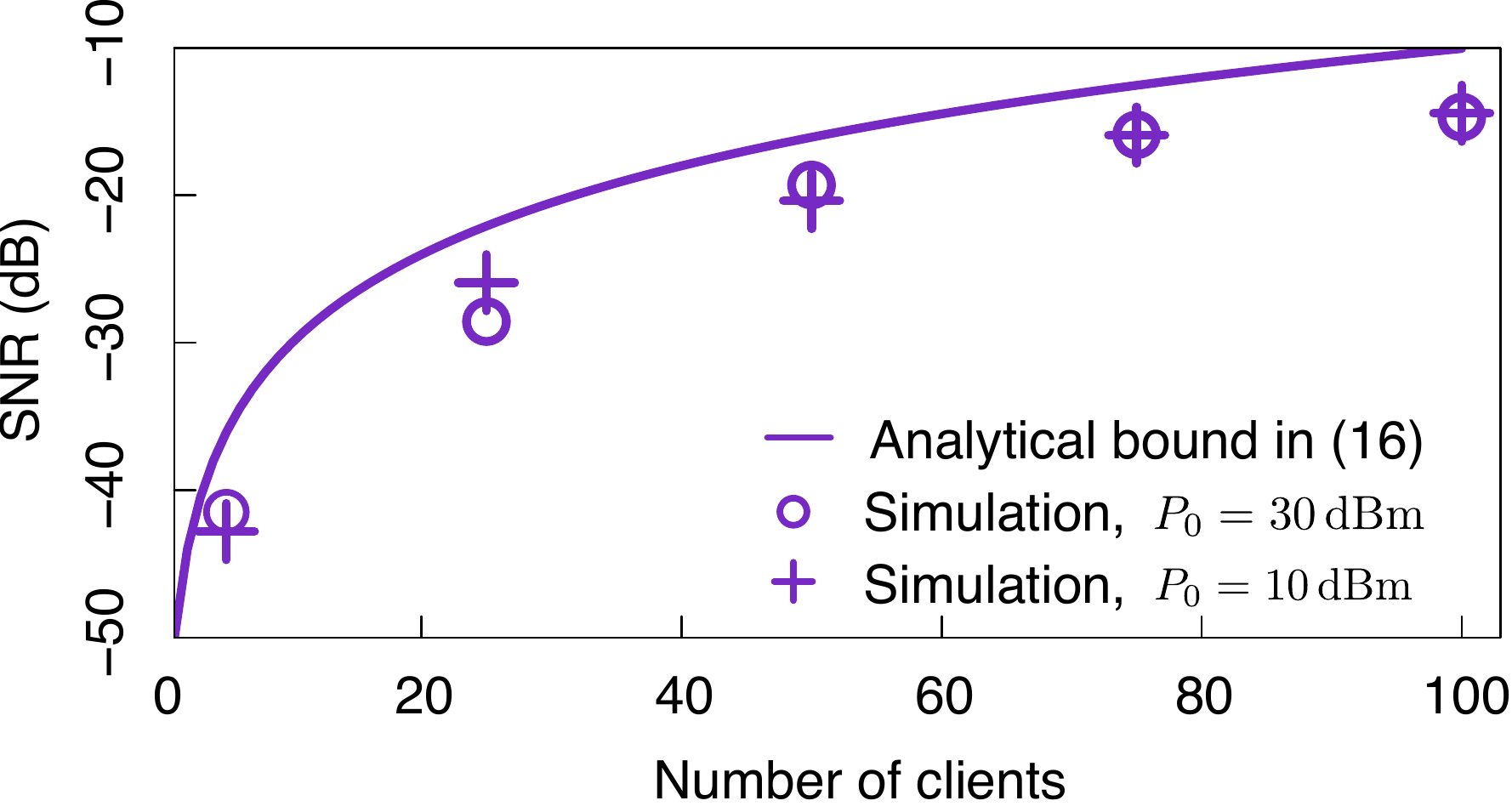}\vspace{-1em}
	\caption{Comparison of analytical results of SNR bound in \eqref{eq:snr_bound_simple} and measured SNR in simulation for $\epsilon = 0.01$.}
	\label{fig:theory_simulation_lower_epsilon}
	\vspace{-.9em}
\end{figure}

\begin{figure}[]
	\centering
	\includegraphics[width=0.75\columnwidth]{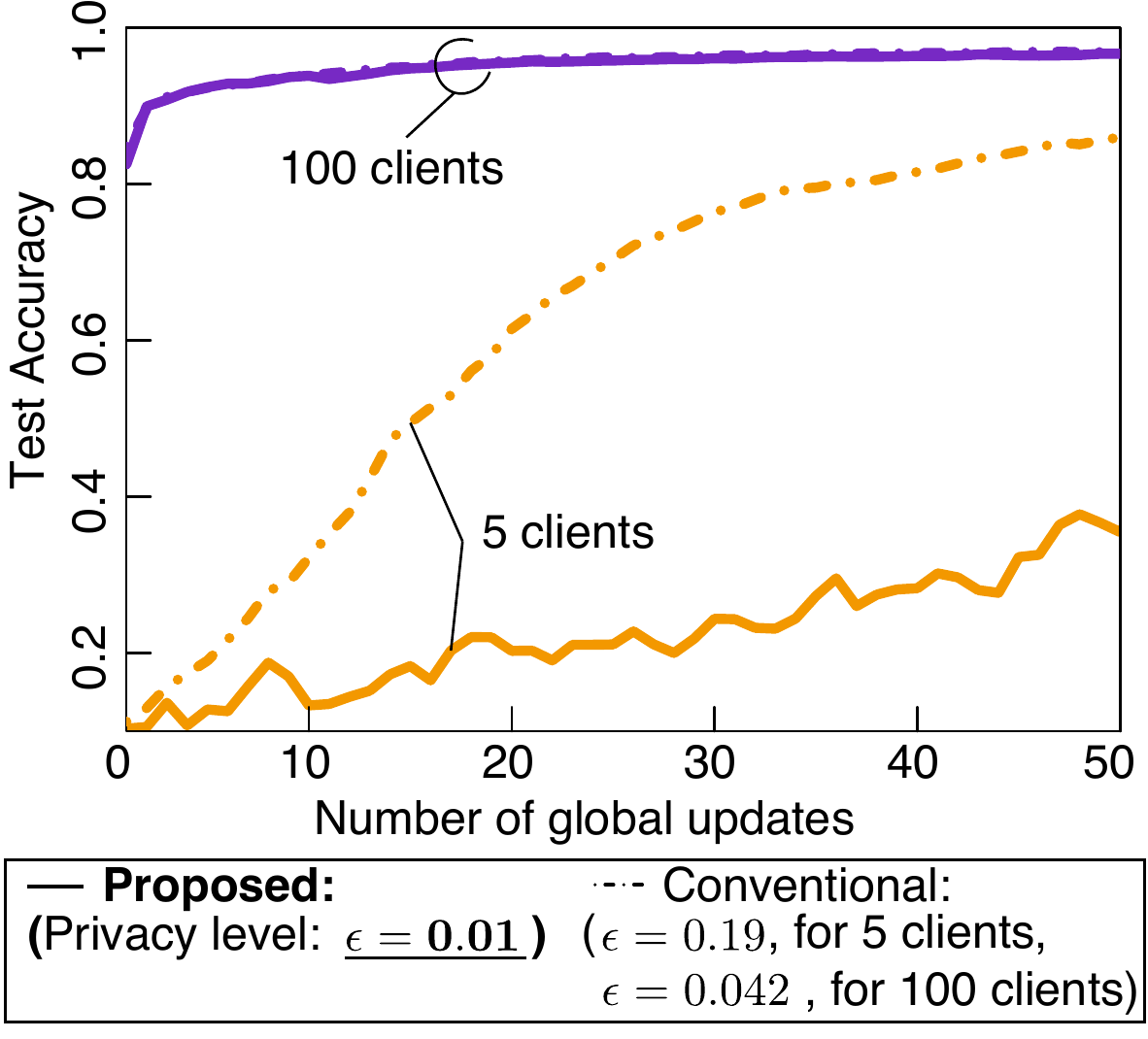}\vspace{-.8em}
	\caption{Training performance for different numbers of clients for target privacy level of $\epsilon=0.01$ and maximum transmit power $P_0 = 10$\,dBm.}
	\label{fig:training_accuracy}
	\vspace{-1.5em}
\end{figure}
\noindent \textbf{Training performance:}
We demonstrate that our designed power control can achieve a higher differential privacy level while exhibiting comparable training performance relative to a conventional power control that sets the transmit power to be maximum in the constraint \eqref{eq:channel_inversion}, i.e., setting the power scaling factor $\rho$ as $\rho_{\mathrm{conv}} \coloneqq \min_{i\in\mathcal{I}}r^{-\alpha}_i |h_i|^2 P_0 / |s_i|^2$\cite{zhu2019broadband}.
Fig.~\ref{fig:training_accuracy} shows the training performance and differential privacy level of our designed power control and conventional power control\footnote{The privacy level in the conventional power control is calculated with Lemma~1 by substituting the average power scaling factor $\mathbb{E}[\rho_{\mathrm{conv}}]$ into \eqref{eq:diff_priv_condition} and then by solving for the privacy level $\epsilon$.} for number of clients equal to $5$ and $100$ in the maximum transmit power $P_0 = 10$\,dBm.
Note that in the designed power control, the target privacy level is set as $\epsilon = 0.01$.
The conventional power control does not achieve the target privacy level in the designed power control, i.e., $\epsilon = 0.01$, which shows the superiority of the designed power control to  conventional one in terms of differential privacy levels.
Regarding the training performance, the designed power control exhibits much poorer performance than the conventional one, owing to the abovementioned tradeoff, i.e., a lower received SNR for achieving the privacy level of $\epsilon=0.01$ when the number of clients is equal to $5$.
Meanwhile, when the number of clients is equal to $100$, the training performance in the designed power control is closer to that in the conventional one, owing to the enhanced received SNR.

\section{Conclusion}
\label{sec:conclusion}
\vspace{-.7em}
We designed AirComp-based FL preserving differential privacy with a desired privacy level to protect local data of clients against privacy inference attacks.
To this end, we considered the use of inherent receiver noises and designed transmit power control to control the level of noise perturbation to the aggregated global model, ensuring a desired privacy level could be achieved.
To gain insight into the challenges to achieve a higher privacy level, we derived a closed-form expression of SNR w.r.t. the privacy level and quantified the tradeoff between these two metrics.
Moreover, the analytical results demonstrate that the number of participating clients is the major factor to enhance the SNR under the tradeoff in particular when a higher privacy level is desired.
These analytical results were verified through numerical evaluations.
The evaluation results also demonstrated the feasibility of the designed power control achieving a performance comparable to that of a conventional power control while achieving a higher differential privacy level.

\vspace{-1em}
\appendix[Proof of Proposition~2]
\vspace{-.5em}
From \eqref{eq:receive_signal2}, the received SNR is given by
\begin{align}
	\label{eq:snr_bound_w_rho}
	\mathit{SNR} &= \frac{\mathbb{E}_{h, s}\!\left[\left|\sqrt{G}\sqrt{\beta}\sum_{i \in \mathcal{I}}\sqrt{\rho^{\star\star}} s_i\right|^2\right]}{\mathbb{E}_{n_0}[|n_0|^2]} \nonumber\\
	&= \frac{G\beta^{-2}\mathbb{E}_{h, s}\!\left[{\rho}^{\star\star}\left|\sum_{i \in \mathcal{I}} s_i\right|^2\right]}{\sigma^2_{\mathrm{n}}}
	\stackrel{(a)}{\leq} \frac{G\beta I^2S^2\mathbb{E}_{h}\!\left[{\rho}^{\star\star}\right]}{\sigma^2_{\mathrm{n}}},
\end{align}
where $h\coloneqq (h_1, \dots, h_I)$ and $s \coloneqq  (s_1, \dots, s_I)$.
The inequality $(a)$ stems from the fact that, in the update clipping in \eqref{eq:clipping}, $\left|\sum_{i\in\mathcal{I}}s_i\right|\leq IS$.

In what follows, we derive $\mathbb{E}_{h}\!\left[{\rho}^{\star\star}\right]$.
To this end, we first prove the following fact: given that $|h_i|^2$ follows an exponential distribution with unit mean, $\min_{i}|h_i|^2r_i^{-{\alpha}}$ follows an exponential distribution with a mean of $1/\sum_{i \in \mathcal{I}}r_i^{\alpha}$.
The proof is obtained by deriving the complementary cumulative distribution function (CCDF) of $\min_{i}|h_i|^2r_i^{-{\alpha}}$.
	Given $x\in\mathbb{R}$, we have
	\begin{align*}
		& \mathbb{P}\!\left(\min_{i\in\mathcal{I}}|h_i|^2r_i^{-{\alpha}}\geq x \right)\\ & = \mathbb{P}\!\left(\bigcap_{i\in\mathcal{I}}|h_i|^2r_i^{-{\alpha}}\geq x \right)
		= \mathbb{P}\!\left(\bigcap_{i\in\mathcal{I}}|h_i|^2\geq x r_i^{{\alpha}}\right)\\
		&= \prod_{i\in\mathcal{I}} \mathbb{P}\!\left(|h_i|^2\geq x r_i^{{\alpha}}\right) 
		= \exp\left(- \sum_{i \in \mathcal{I}}r^{\alpha}_i x\right),
	\end{align*}
	which is exactly the CCDF of an exponential distribution with mean $1/\sum_{i\in\mathcal{I}}r_i^{\alpha}$.

Let $\frac{\sigma^2_{\mathrm{n}}}{4\,G\beta P_0}\frac{\epsilon^2}{\ln(1.25/\delta)}$ be denoted as $g_{\mathrm{th}}$.
From the fact that $g\coloneqq \min_{i\in\mathcal{I}} r^{-\alpha}_i |h_i|^2$ follows the exponential distribution with mean $1/\sum_{i\in\mathcal{I}}r^{\alpha}_i$, we obtain $\mathbb{E}_{h}\!\left[{\rho}^{\star\star}\right]$ as follows:
\begin{align}
	\label{eq:expectation_rho}
	\mathbb{E}_{h}\!\left[{\rho}^{\star\star}\right] 
	&= \int_{0}^{g_{\mathrm{th}}}\,\rho^{\star\star}\, p(g)\,\mathrm{d}g
	+ \int_{g_{\mathrm{th}}}^{\infty}\,\rho^{\star\star}\, p(g)\,\mathrm{d}g\nonumber\\
	&= \int_{0}^{g_{\mathrm{th}}}\,\frac{gP_0}{S^2}\, \sum_{i \in \mathcal{I}}r^{\alpha}_i\exp{\left(- g\sum_{i \in \mathcal{I}}r^{\alpha}_i\right)}\,\mathrm{d}g\nonumber\\
	&\quad + \int_{g_{\mathrm{th}}}^{\infty}\,\frac{g_{\mathrm{th}}P_0}{S^2}\, \sum_{i \in \mathcal{I}}r^{\alpha}_i\exp{\left(- g\sum_{i \in \mathcal{I}}r^{\alpha}_i\right)}\,\mathrm{d}g\nonumber\\
	& = \frac{P_0}{S^2 \sum_{i \in \mathcal{I}}r^{\alpha}_i}\left(1 - \exp\!\left({-g_{\mathrm{th}}\sum_{i \in \mathcal{I}}r^{\alpha}_i}\right)\right).
\end{align} 
Substituting \eqref{eq:expectation_rho} into \eqref{eq:snr_bound_w_rho}, we obtain \eqref{eq:snr_bound}.

\section*{Acknowledgment}
This work was supported in part by JSPS KAKENHI Grant Numbers JP17H03266 and JP18H01442.

\bibliographystyle{IEEEtran}
\bibliography{IEEEabrv,fl}

\end{document}